\documentclass[a4paper, UKenglish, cleveref, autoref, thm-restate]{lipics-v2021}
\usepackage{graphicx, float, amsmath, amssymb, amsthm, url, hyperref, xurl}
\usepackage[linesnumbered,ruled,noend]{algorithm2e}
\usepackage[T1]{fontenc}
\usepackage[utf8]{inputenc}
\usepackage{tikz}
\usepackage{mathtools}

\nolinenumbers

\DeclareMathOperator{\DAG}{\mathbf{DAG}}
\DeclareMathOperator{\Set}{\mathbf{Set}}
\newcommand\ggf[1]{\mathbf{#1}}
\newcommand\nbv[1]{v(#1)}
\newcommand\nbe[1]{e(#1)}
\newcommand\nbs[1]{s(#1)}

\DeclareMathOperator\aprod{\raisebox{-.1\height}{\begin{tikzpicture}
  \draw[->] (-.4em,0) -- (.4em,0);
  \draw (0,0) circle (.4em);
\end{tikzpicture}}}
\newcommand\Seq[1]{\ensuremath{\textsc{Seq}(#1)}}

\newcommand\A{\mathcal{A}}
\newcommand\B{\mathcal{B}}
\newcommand\C{\mathcal{C}}

\newcommand\G{\mathcal{G}}

\newcommand\PP{\mathbb P}

\title{\textbf{Boltzmann sampling and optimal exact-size sampling for directed acyclic graphs}}
\author{Wojciech Gabryelski}{Department of Fundamentals of Computer Science, Wroc{\l}aw University of Science and Technology, Poland}{wojciech.gabryelski@pwr.edu.pl}{https://orcid.org/0009-0008-0333-167X}{}
\author{Zbigniew {Go{\l}\k{e}biewski}}{Department of Fundamentals of Computer Science, Wroc{\l}aw University of Science and Technology, Poland}{zbigniew.golebiewski@pwr.edu.pl}{https://orcid.org/0000-0002-0186-5263}{}
\author{Martin Pépin}{Université Caen Normandie, ENSICAEN, CNRS, Normandie Univ, GREYC UMR 6072, F-14000
Caen, France}{martin.pepin@unicaen.fr}{https://orcid.org/0000-0003-1892-3017}{}
\authorrunning{W. Gabryelski, Z. {Go{\l}\k{e}biewski}, M. P\'epin}
\Copyright{Wojciech Gabryelski, Zbigniew {Go{\l}\k{e}biewski} and Martin P\'epin}

\funding{\emph{W. Gabryelski, Z. {Go{\l}\k{e}biewski}, M. P\'epin:} This research was funded in whole or in part by National Science Centre, Poland, grant OPUS-25 no 2023/49/B/ST6/02517 and by the \textsc{anr-fwf} project PAnDAG ANR-23-CE48-0014-01.}



\begin{document}

\maketitle

\begin{abstract}
    We propose two efficient algorithms for generating uniform random directed acyclic graphs, including an asymptotically optimal exact-size sampler that performs $\frac{n^2}{2} + o(n^2)$ operations and requests to a random generator. This was achieved by extending the Boltzmann model for graphical generating functions and by using various decompositions of directed acyclic graphs.
    The presented samplers improve upon the state-of-the-art algorithms in terms of theoretical complexity and offer a significant speed-up in practice.
\end{abstract}

\section{Introduction}

Random generation has many applications in various scientific areas. It is, for instance, used in physics for the simulation of phenomena such as the Ising
model~\cite{Velenik2009} or chord diagrams (which are linked with quantum field
theory)~\cite{CYZ2016}, in biology for the study of RNA secondary
structures~\cite{Ponty2006}, in software engineering for automated
testing \textit{a la} QuickCheck~\cite{CH2000} and fuzzing~\cite{peach}, in
algorithmics for computing the volume of convex polytopes~\cite{LV2006} or
the permanent of matrices~\cite{JS1989} using randomised algorithms, or in
mathematics and theoretical computer science as an experimentation
tool~\cite{BM2022}.

Due to the widespread use of random generation, generic approaches and
frameworks for the design of samplers have been developed. One of the earliest
such framework is that of the so-called ``recursive-method'' --- introduced
in~\cite{NW1978} and later systematised in~\cite{FZV1994} --- where the idea is
to rely on counting sequences to guide the various random choices made by the
sampling algorithm.
Another frequently used approach for random generation is the Markov chain method~\cite{habib1998probabilistic}, which involves
constructing a Markov chain whose states are the desired structures (\textit{e.g.}\ graphs, permutations, trees) and running a random walk on this state space until it converges to a stationary distribution, typically designed to be uniform.
Another important approach, upon which we build in this work, is the Boltzmann method described
in~\cite{BoltzSamp2004}. This method, which has undergone significant
improvements and generalisations since its
inception~\cite{FFP2007,PSS2008,BRS2012,BLR2015,BBD2018}, enables
the automatic derivation of efficient samplers from combinatorial specifications of
a size \emph{close} to a target size~$n$ and guarantees that two objects of
the same size are equiprobable.

In this paper, we focus on the particular problem of the uniform random
generation of Directed Acyclic Graphs (DAGs).
DAGs are a ubiquitous data structure in computer science. They appear naturally
as a result of tree compaction, for instance for the compression of XML
documents~\cite{BLMN2015}, or as a means of representing partial orders~\cite{KS2020},
typically in scheduling problems~\cite{CEH2019,CMPTVW2010}.
The combinatorial study of DAGs dates back at least to the early 1970s with
the work of Robinson~\cite{robinson1970enumeration,Robinson1973} and
Stanley~\cite{STANLEY1973171}.
The problem of efficient uniform DAG generation is more recent however. A first
algorithm was given in 2001 in~\cite{MELANCON2001202}, based on the Markov chain
approach, and motivated by information visualisation applications.
Since then, the problem of efficient DAG sampling has drawn significant
attention, particularly in statistics for inferring the structure of Bayesian
networks~\cite{KF2009}. A fruitful line of
work~\cite{Kuipers_2013,KM2017,TVK2020,KSM2022} pushed the limits of DAG
sampling from the~$O(n^5 \ln n)$ complexity of the initial
design~\cite{MELANCON2001202}%
\footnote{The original article claims a~$O(n^4)$ complexity that was later
disproved by~\cite{Kuipers_2013}.}
to a~$O(n^2)$ time complexity with quadratic pre-processing, thus achieving
near-optimal complexity.
Notably, this last algorithm is based on ideas from the recursive
method, although it was initially thought that this approach was slower than
Markov chain based techniques.

The first contribution of the present paper is to \textbf{extend the framework
of Boltzmann sampling to families of directed graphs}
and to demonstrate its effectiveness by applying it to DAG generation.
This work opens the way for largely automating the process of graph generation
without sacrificing performance, as we shall demonstrate.
The second contribution of this paper, made possible by our extension of the Boltzmann framework, is to close the gap between existing
approaches and the theoretically optimal complexity for DAG generation by providing
an asymptotically \textbf{optimal exact-size sampling algorithm}.
By optimal, we mean that our sampler consumes~$\frac{n^2}{2} + o(n^2)$ random bits on average, which is asymptotic to the entropy of the uniform distribution over size-$n$ DAGs.
Contrary to the state-of-the-art sampler from~\cite{TVK2020}, our algorithm does not require any preprocessing.
Furthermore, in terms of memory accesses, our algorithm can generate DAGs of size~$n$ by
performing a first~$O(n)$-time rejection phase, followed by a single pass over the adjacency matrix of the output, which induces only~$\frac{n^2}2 + O(n)$ memory accesses.

The paper is structured as follows.
In Section~\ref{sec:def_and_prop}, we recall the notion of the graphic generating function from~\cite{Robinson1973,Panafieu2019SymbolicMA} as well as earlier results on DAGs, and introduce a generalisation of the Boltzmann model for digraph families.
Then, in Section~\ref{sec:freesamplers}, we present two possible approaches to implementing Boltzmann samplers for DAGs. In addition to the algorithmic results, these two approaches offer complementary perspectives on the structure of random DAGs with different implications, which is why we chose to present both.
Building on our samplers from the previous section, we demonstrate in Section~\ref{sec:exactsize} how to leverage the Boltzmann framework to achieve optimal exact-size sampling for uniform DAGs.

In practice, the C/C++ implementation of our algorithms performs several orders of magnitude faster than the state-of-the-art sampler available at~\cite{ModularSamplerImpl} (also written in C++). Whereas their implementation samples a uniform DAG of size~$n=4096$ in around~3s, our fastest algorithm does so in about~$20$ms and can reach size~$n=200000$ in about~$2$s on a consumer-grade laptop.
Our implementations are available online at~\url{https://osf.io/g4xk8/overview?view_only=a7781b4b2be54b61b0087ee02a4fee5e} for others to reproduce our experiments, and a more thorough benchmark will be conducted in the near future.

\section{Definitions and basic properties}\label{sec:def_and_prop}

Let $\mathbb{R}$ and $\mathbb{Z}$ denote the sets of real numbers and integers, respectively. For $a, b \in \mathbb{Z}$ such that $a<b$, let $[a..b] = \{a,a+1,\dotsc,b\}$.
For a set $X$ and some functions $f,g: X \to \mathbb{R}$, we write $f \propto g$ if $f$ is proportional to $g$, which means that there exists a constant $c \neq 0$ such that $f(x) = c \cdot g(x)$ for any $x\in X$.
For a generating function $f(z)$, we denote the operation of extracting the coefficient of $z^n$ in the formal power series $f(z) = \sum_{n} f_n z^n$ by $[z^n] f(z)$.
By $\partial_x f(x,y)$, we denote the partial derivative of a function $f(x,y)$ with respect to $x$.

\subsection{Graphic generating functions}

For our purposes, we need to use an uncommon type of generating function called a \emph{graphic generating function} (GGFs). For any class $\G$ consisting of graphs, the graphic generating function of~$\G$ is given by the formula
\[
\ggf{G}(z,w,u)=\sum_{G\in\G} \frac{z^{\nbv G}w^{\nbe G}u^{\nbs G}}{(1+w)^{\binom{\nbv G}{2}} \nbv G!}~,
\]
where $\nbv G$ is the number of vertices of the graph $G$, $\nbe G$ is its number of edges, and $\nbs G$ is its number of sources. Whenever the variable $u$ is not of interest, we write $\ggf{G}(z,w)=\ggf{G}(z,w,1)$.

In particular, by $\DAG(z,w,u)$ we denote the graphic generating function of labelled directed acyclic graphs. GGFs were first introduced by Robinson in~\cite{Robinson1973} under the name of \emph{special generating functions} in order to derive exact and asymptotic enumeration formulas for DAGs. The more explicit name of \emph{graphic generating functions} dates back to at least 1995 in the work of Gessel~\cite{GESSEL1995257}, who used them to derive more recurrence formulas for DAGs.
What makes GGFs useful for the enumeration of graphical objects is that their product encodes the so-called ``arrow-product''.

\begin{definition}\label{def:arrow_product}
    The arrow product of two classes $\A$ and $\B$ of digraphs, denoted by~$\A \aprod \B$, is the class of all digraphs resulting from connecting any pair of graphs $(a,b)\in\A\times\B$ with an arbitrary number of directed edges from the vertices of graph $a$ to the vertices of graph $b$.
\end{definition}
An illustration of the arrow product is given in Figure~\ref{fig:arrowprod}.
This operation can be nicely expressed in the framework of graphic generating functions since
\begin{equation}\label{eq:arrowprod}
    (\ggf{A}\aprod\ggf{B})(z,w)=\ggf{A}(z,w)\cdot \ggf{B}(z,w)~.
\end{equation}
A proof of this property, as well as many applications of GGFs are available in~\cite{Panafieu2019SymbolicMA}.
\begin{figure}
   \centering
   \includegraphics[scale=1]{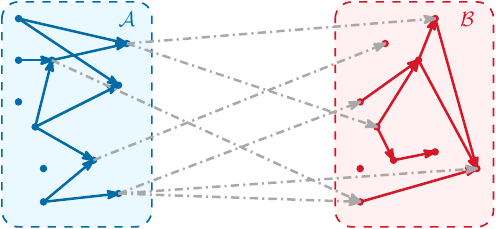}
   \caption{Illustration of the arrow product: the two graphs are drawn for some classes of digraphs~$\A$ and~$\B$ and the grey dotted edges are a possible set of edges from the~$\A$ component to the~$\B$ component introduced by the arrow product. Labels are omitted for the sake of clarity.}%
   \label{fig:arrowprod}
\end{figure}

As mentioned above, the graphic generating function of directed acyclic graphs is already known from the works of Robinson~\cite{Robinson1973}, but it was also independently obtained by Stanley~\cite{STANLEY1973171} using a different approach.
It was later used to study various parameters of DAGs, notably in~\cite{GESSEL1996253, Panafieu2019SymbolicMA, DdRRW2024}.
Let $\Set(z,w)$ be the graphic generating function of arcless graphs, given by the formula
\begin{equation}\label{eq:set}
    \Set(z,w)=\sum_{n\geq0}\frac{z^n}{(1+w)^{\binom{n}{2}}n!}~.
\end{equation}
The aforementioned studies have shown that
\begin{equation}\label{eq:dag}
    \DAG(z,w)=\frac{1}{\Set(-z,w)}
    \qquad
    \text{and}
    \qquad
    \DAG(z,w,u)=\frac{\Set((u-1)z,w)}{\Set(-z,w)}.
\end{equation}
An interesting property of the~$\Set$ function that we will use on several occasions in this paper is that
\begin{equation}\label{eq:set_first_derivative}
    \partial_z \Set(z, w) = \Set\left(\frac{z}{1+w}, w\right)\!.
\end{equation}

\subsection{A graphic Boltzmann model}

In~\cite{BoltzSamp2004}, the authors defined two variants of the Boltzmann model: the ordinary and exponential versions (corresponding to combinatorial classes with ordinary and exponential generating functions, respectively).
Since we are dealing with graphic generating functions, we need to introduce a natural extension of the Boltzmann model for digraphs.

\begin{definition}[Graphic Boltzmann model]\label{def:boltzmann}
    Given a class~$\mathcal G$ of digraphs, the graphic Boltzmann model on~$\mathcal G$ is the probability distribution assigned to any $G \in \mathcal{G}$ given by
    \begin{equation*}
        \PP_{\mathcal G, z,w,u}[G] = \frac{z^{\nbv G}w^{\nbe G}u^{\nbs G}}{(1+w)^{\binom{\nbv G}{2}}\nbv G!\ggf{G}(z,w,u)}\cdot
    \end{equation*}
    Whenever it is clear from the context, we will omit the~$\mathcal G$ subscript in order to make the notation lighter. The same applies to the variable $u$ when it is not of interest (we set $u=1$). We will also omit the term ``graphic'' in the rest of the paper since this is the only variant of the model that we use here.
\end{definition}

Our goal is to find a Boltzmann sampler $\Gamma\mathcal{DAG}(z,w)$, which is an
algorithm that, for given parameters $z$ and $w$, produces a random labelled
directed acyclic graph $G$ according to the corresponding Boltzmann model.
The parameters $z$ and $w$ influence the expected number of vertices and edges
in the generated DAG. Graphs with the same number of vertices and edges are
equally probable. Note that for $w=1$, the probability of generating a given DAG $G$ is
\[
\PP_{z,1}[G]=\frac{z^{\nbv G}}{2^{\binom{\nbv G}{2}}\nbv G!\DAG(z,1)}~,
\]
therefore all graphs with the same number of vertices are equally probable.

The graphic Boltzmann model introduced above has composition properties similar to those of the more classical ordinary and exponential Boltzmann models.
However, providing a comprehensive overview of graphic Boltzmann samplers is beyond the scope of this paper, as we will focus only on DAG sampling.

\subsection{A useful property of graphic generating functions}

Consider labelled digraphs that contain a fixed set of vertices $W$, and let $\C$ be any set of digraphs whose vertex set is exactly $W$.
We will refer to $\C$ as {\it possible} configurations.
Now, let $\G$ be a class of {\it possible} digraphs, that is:
\begin{itemize}
    \item for every $G \in \G$, the induced subgraph $G[W]$ is isomorphic to some $C \in \C$,
    \item let $C = (W, F) \in \C$ and $G = (V \dot\cup W, E \dot\cup F) \in \G$, then for all $C' = (W, F') \in \C \setminus \{C\}$, all graphs $G' = (V \dot\cup W, E \dot\cup F')$ also belong to $\G$.
\end{itemize}
Analogously, let  $\hat{\C} \subseteq \C$ be a set of permitted configurations and $\hat{\G}$ be a class of permitted digraphs in which only permitted configurations appear as induced subgraphs.
Then, according to the Boltzmann model, the probability of drawing a permitted graph that belongs to $\hat{\G}$ from all possible graphs $\G$ is equal to
\begin{equation}\label{eq:ggf_trick}
    \PP_{\G,z,w}[\hat{\G}]=\frac{\ggf{\hat{G}}(z,w)}{\ggf{G}(z,w)}=\frac{
    \sum_{\hat{C} \in \hat{\C}} w^{e(\hat{C})}
    }{
    \sum_{C \in \C} w^{e(C)}
    }~.
\end{equation}
For example, let $W$ be a fixed set of $n$ vertices, $\C$ contain all digraphs with the vertex set $W$, and we allow only one specific digraph $C$ with $j$ edges (that is, $\hat{\C} = \{C\}$). Then, the probability of sampling the allowed graph from the possible graphs $\G$ is given by $w^j / (1+w)^{n (n-1)}$.

\section{Two complementary approaches to Boltzmann sampling}\label{sec:freesamplers}

In this section, we present two distinct solutions to the problem of implementing a Boltzmann sampler for DAGs. One of our solutions is based on the notion of root-layering that Robinson used in \cite{robinson1970enumeration} to establish the first recurrence formulas for DAGs. This idea is natural in the context of DAGs, but requires us to design an \textit{ad-hoc} procedure in order to handle the layers in the Boltzmann model.
Our second solution relies on a new recursive decomposition of DAGs, reminiscent of peeling processes on maps, which avoids the use of inclusion-exclusion and thus allows to use more standard Boltzmann sampling techniques.

\subsection{Boltzmann sampling based on root-layerings}

\subsubsection{Unique decomposition of DAGs: root-layering}

In this section, we present a directed acyclic graph decomposition, which is a crucial observation for the construction of our Boltzmann sampler.
\begin{definition}
    A \emph{root-layering} $(V_1,\dotsc,V_k)$ of a directed acyclic graph $G=(V,E)$ is an ordered partition of the set of vertices $V$ such that for $i\in[1..k]$, the set $V_i$ consists of the sources of the graph resulting from the removal from the graph $G$ of all vertices belonging to the previous sets $V_1,\dotsc,V_{i-1}$.
\end{definition}
The example of root-layering is shown in Fig.~\ref{fig:graphs}.
\begin{figure}[H]
  \centering
  \includegraphics[scale=1]{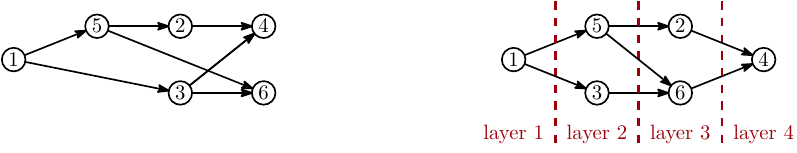}
  \caption{A labelled DAG with~$6$ vertices and its root-layering on the right.
  Here the layers are~$\{1\}$, $\{3, 5\}$, $\{2, 6\}$, and~$\{4\}$.}%
  \label{fig:graphs}
\end{figure}
\begin{theorem}\label{thm:root_layering}
    Given a directed acyclic graph $G=(V,E)$, let $(V_1,\dotsc,V_k)$ be the root-layering of the graph $G$. Then $(V_1,\dotsc,V_k)$ is the only tuple $(U_1,\dotsc,U_l)$ of non-empty disjoint sets summing up to $V$ that satisfies
    \begin{enumerate}
        \item (disallowed backward edges) $(\forall i,j\in[1..l])(i\leq j\implies(U_j\times U_i)\cap E=\emptyset)$,
        \item (obligatory forward edges) $(\forall i\in[1..l-1])(\forall v\in U_{i+1})(\exists u\in U_i)((u,v)\in E)$.
    \end{enumerate}
    Furthermore, for given disjoint sets of vertices $V_1,\dotsc,V_k$ such that $V=V_1\cup\ldots\cup V_k$, and any set of edges $E\subseteq V\times V$, if the graph $G=(V,E)$ satisfies property 1, then $G$ is a directed acyclic graph.
\end{theorem}
\begin{proof}
    For now, let us focus on the first part of the theorem. For $i\in[1..k]$, let $G_i$ be the subgraph of the graph $G$ induced by the set of vertices $V_i\cup\ldots\cup V_k$. Then $V_i$ is the set of sources of the graph $G_i$, thus for every $v\in V_i$ there is no vertex $u\in V_i\cup\ldots\cup V_k$ such that $(u,v)\in E$, which proves property 1. Fix $i\in[1..k-1]$ and some vertex $v\in V_{i+1}$. This vertex is the source of the graph $G_{i+1}$, but it is not the source of the graph $G_i$. Hence, there must exist some vertex $u\in V_i\cup\ldots\cup V_k$ such that $(u,v)\in E$, but this vertex does not belong to the set $V_{i+1}\cup\ldots\cup V_k$, so $u\in V_i$ must hold, which proves property 2.

    Now let $(U_1,\dotsc,U_l)$ be some tuple of non-empty disjoint sets summing up to $V$, which is not the root-layering of the graph $G$. For $i\in[1..l]$, let $G'_i$ be the subgraph of the graph $G$ induced by the set of vertices $U_i\cup\ldots\cup U_l$. Then there exists $i$ such that some source $v$ of the graph $G_i'$ does not belong to the set $U_i$ or there exists $v\in U_i$ such that $v$ is not a source of the graph $G_i'$. In the first case, there exists $j\in[i+1..l]$ such that $v\in U_j$. Since $v$ is the source of the graph $G_i'$ and $j-1\geq i$ (so the vertices of $U_{j-1}$ are also the vertices of the graph $G_i'$), there is no vertex $u\in U_{j-1}$ such that $(u,v)\in E$, so property 2 is not satisfied. In the second case, since the vertex $v\in U_i$ is not a source of the graph $G_i'$, there must exist some vertex $u\in U_i\cup\ldots\cup U_l$ such that $(u,v)\in E$, so property 1 does not hold.

    Let us move on to the second part of the theorem. Assume that property 1 holds. Fix $i\in[1..k]$ and $v_0\in V_i$. Let $v_0\to v_1\to\ldots\to v_j$ be a path in the graph $G$, and let $i_0,i_1,\dotsc,i_j$ be indices such that $v_0\in V_{i_0},v_1\in V_{i_1},\dotsc,v_j\in V_{i_j}$. Then, for property 1 to hold, we must have $i_0<i_1<\ldots<i_j$, which implies that $i_0\neq i_j$. This means that $v_0\neq v_j$ and the path $v_0\to v_1\to\ldots\to v_j$ is not a cycle. Therefore, no path in the graph $G$ forms a cycle, and thus $G$ is a directed acyclic graph.
\end{proof}

\subsubsection{Boltzmann sampler}

From the previous section, we can see that DAGs can be drawn as follows:
\begin{enumerate}
    \item generate layers of vertices,
    \item draw edges satisfying properties 1 and 2 from Theorem~\ref{thm:root_layering}, 
    \item assign random labels to the vertices.
\end{enumerate}
In this way, we will generate a specific DAG with the root-layering consisting of the layers generated in the first step.

To generate graphs according to the Boltzmann model, we need to derive a proper probability distribution for steps 1 and 2 of the above general schema.
The probability distribution of the sizes of consecutive sets in the root-layering will be shown in Lemmas~\ref{lem:root_layering_homogenous_markov_chain}-\ref{lem:root_layering_transition_matrix} and the probability distribution of the edges given the root-layering will be shown in Lemma~\ref{lem:root_layering_edges_distribution}.

First, let us determine the probability distributions of the sizes of sets in the root-layering of a random DAG. Let $N_1,N_2,\dotsc$ be random variables that denote the sizes of successive sets $V_1,V_2,\dotsc$ in the root-layering of a DAG drawn from the distribution given by the Boltzmann model. Assume that we have already drawn $k$ sets of sizes $n_1,\dotsc,n_k$. We need to determine the probability distribution of the size of the next set in the root-layering $N_{k+1}$ conditioned on $N_1=n_1,\dotsc,N_k=n_k$.

\begin{lemma}\label{lem:root_layering_homogenous_markov_chain}
    The process of generating sizes of consecutive sets in the root-layering is a time-homogeneous Markov chain, that is
    \begin{equation*}
        \PP_{z,w}[N_{k+1}=n_{k+1}|N_1=n_1,\dotsc,N_k=n_k]=\PP[N_2=n_{k+1}|N_1=n_k]~.
    \end{equation*}
\end{lemma}
\begin{proof}
    For any DAG $G$, let $rl(G)$ be the sequence of the sizes of successive sets in the root-layering of $G$. For any sequence $\sigma=(n_1,\dotsc,n_k)$, let $\sigma\preceq rl(G)$ mean that the first $k$ sets in the root-layering of graph $G$ have sizes $n_1,\dotsc,n_k$. According to the Boltzmann model given in Definition~\ref{def:boltzmann}, the probability that the set $V_{k+1}$ has size $n_{k+1}$ provided that the previous sets have sizes $n_1,\dotsc,n_k$ is proportional to the graphic generating function of the class of all DAGs such that $(n_1,\dotsc,n_k,n_{k+1})\preceq rl(G)$. Each such DAG can be divided into three parts: the subgraph induced by the vertices $V_1\cup\ldots\cup V_{k-1}$, the subgraph induced by the remaining vertices, which has $n_k$ sources, and the edges between vertices $V_1\cup\ldots\cup V_{k-1}$ and the remaining vertices. Note that from Theorem~\ref{thm:root_layering} we know that all these edges are optional except the edges between sets $V_{k-1}$ and $V_k$, where every vertex from $V_k$ needs to be connected by an edge with at least one vertex from $V_{k-1}$. Hence, the class of all DAGs $G$ such that $(n_1,\dotsc,n_k,n_{k+1})\preceq rl(G)$ can be obtained by performing an arrow product on the class of all DAGs $G$ such that $(n_1,\dotsc,n_{k-1})=rl(G)$ and the class of all DAGs $G$ such that $(n_k,n_{k+1})\preceq rl(G)$, and discarding forbidden combinations of edges between sets $V_{k-1}$ and $V_k$. From \eqref{eq:ggf_trick} we can see that the generating function of this class can be obtained by multiplying the graphic generating function of the arrow product by the term
    $\left(\frac{(1+w)^{n_{k-1}}-1}{(1+w)^{n_{k-1}}}\right)^{n_k}$,
    since it corresponds to allowing only non-empty combinations of edges between $V_{k-1}$ and $v$ for each vertex $v\in V_k$. From the formula~\eqref{eq:arrowprod} for the graphic generating function of the arrow product, we obtain
    \begin{align*}
        &\PP_{z,w}[N_{k+1}=n_{k+1}|N_1=n_1,\dotsc,N_k=n_k]\propto\sum_{\substack{G\in\mathcal{DAG} \\ (n_1,\dotsc,n_k,n_{k+1})\preceq rl(G)}}\ggf{G}(z,w)\\
        &=\left(\frac{(1+w)^{n_{k-1}}-1}{(1+w)^{n_{k-1}}}\right)^{n_k}\left(\sum_{\substack{G\in\mathcal{DAG} \\ (n_1,\dotsc,n_{k-1})=rl(G)}}\ggf{G}(z,w)\right)\left(\sum_{\substack{G\in\mathcal{DAG} \\ (n_k,n_{k+1}) \preceq rl(G)}}\ggf{G}(z,w)\right)~,
    \end{align*}
    where $\ggf{G}(z,w) = \frac{z^{\nbv G}w^{\nbe G}}{(1+w)^{\binom{\nbv G}{2}} \nbv G!}$ denotes the graphic generating function of the class consisting of a single graph $G$. Note that only the last factor in the above expression depends on $n_{k+1}$, and the previous factors can be treated as constants. Thus
    \[
    \PP_{z,w}[N_{k+1}=n_{k+1}|N_1=n_1,\dotsc,N_k=n_k]\propto\sum_{\substack{G\in\mathcal{DAG} \\ (n_k,n_{k+1})\preceq rl(G)}}\frac{z^{\nbv G}w^{\nbe G}}{(1+w)^{\binom{\nbv G}{2}}\nbv G!}~,
    \]
    so the probability distribution does not depend on $n_1,\dotsc,n_{k-1}$ nor $k$, it only depends on $n_k$.
\end{proof}

Before determining the probability distribution of the sizes of consecutive sets in the root-layering, we will calculate the probability distribution of the number of sources in a random DAG.
\begin{lemma}[Probability distribution of sources]\label{lem:root_layering_sources_distribution}
    The probability distribution of the number of sources in a random directed acyclic graph drawn according to the Boltzmann model is given by the formula
    \begin{equation*}
        \PP_{z,w}[N_1=n]=\frac{z^n}{(1+w)^{\binom{n}{2}}n!}\Set\left(-(1+w)^{-n}z,w\right)~.
    \end{equation*}
\end{lemma}
\begin{proof}
 The probability distribution of the number of sources can be obtained using formulas \eqref{eq:set} and \eqref{eq:dag} and the following derivation:
 \[
    \PP_{z,w}[N_1=n] = \frac{[u^n]\DAG(z,w,u)}{\DAG(z,w)} = [u^n]\Set((u-1)z,w) = \frac{z^n \Set\left(-(1+w)^{-n}z,w\right)}{(1+w)^{\binom{n}{2}}n!}.\qedhere
 \]
\end{proof}

To complete the process of generating the root-layering, we need to determine the transition matrix of the Markov chain.
\begin{lemma}[Transition matrix]\label{lem:root_layering_transition_matrix}
    The entries of the transition matrix describing the process of generating consecutive set sizes in the root-layering are given by the following formula
    \begin{equation*}
        \PP_{z,w}[N_2=n_2|N_1=n_1]=\frac{\left((1-(1+w)^{-n_1})z\right)^{n_2}}{(1+w)^{\binom{n_2}{2}}n_2!}\frac{\Set(-(1+w)^{-n_2}z,w)}{\Set(-(1+w)^{-n_1}z,w)}~.
    \end{equation*}
\end{lemma}
\begin{proof}
    Let us note that the class of DAGs $G$ following $(n_1,n_2)\preceq rl(G)$ can be obtained by performing an arrow product on the class consisting of one empty graph with $n_1$ vertices and the class of all DAGs with $n_2$ sources, and then replacing all combinations of edges between the first two sets in a root-layering with all non-empty combinations of these edges. With this observation, the transition matrix, which is the probability distribution of $N_{k+1}$ conditioned on the value of $N_k$ for any $k\geq1$, can be calculated using the result of Lemma~\ref{lem:root_layering_sources_distribution}.
    \begin{align*}
        \PP_{z,w}[N_2=n_2|N_1=n_1] &= \frac{\PP_{z,w}[N_1=n_1,N_2=n_2]}{\PP_{z,w}[N_1=n_1]} \\
        &= \frac{\left(\frac{(1+w)^{n_1}-1}{(1+w)^{n_1}}\right)^{n_2}\frac{z^{n_1}}{(1+w)^{\binom{n_1}{2}}n_1!}[u^{n_2}]\DAG(z,w,u)}{\PP_{z,w}[N_1=n_1]\DAG(z,w)}\\
        &=\frac{\left((1-(1+w)^{-n_1})z\right)^{n_2}}{(1+w)^{\binom{n_2}{2}}n_2!}\frac{\Set(-(1+w)^{-n_2}z,w)}{\Set(-(1+w)^{-n_1}z,w)}~.\qedhere
    \end{align*}
\end{proof}

Assume that the root-layering has already been generated. Now we need to draw edges so that the properties presented in Theorem~\ref{thm:root_layering} are satisfied.
\begin{lemma}[Probability distribution of edges]\label{lem:root_layering_edges_distribution}
    Each edge between two non-consecutive sets in the root-layering occurs independently of other edges with probability $\frac{w}{1+w}$.
    For each two consecutive sets $V_k,V_{k+1}$ in the root-layering, the edges between $V_k$ and any vertex $v\in V_{k+1}$ can be drawn with probability $\frac{w}{1+w}$ until a combination of more than one edge is drawn.
\end{lemma}
\begin{proof}
    Assume that we may already have drawn some edges in the generated graph. Let $\G$ be the class of all graphs that may be generated from now on. Let $e$ be a pair of vertices from non-consecutive sets in the root-layering, such that we have not yet decided whether an edge between those vertices should occur in the generated graph. Let $\hat{\G}$ be the class of graphs from $\G$ that contains the edge $e$. From \eqref{eq:ggf_trick}, we know that the probability that $e$ will occur in the generated graph is $\frac{\ggf{\hat{G}}(z,w)}{\ggf{G}(z,w)}=\frac{w}{1+w}$, which means that it is independent of the edges that have already been generated. Hence, each such edge can be generated independently with a probability $\frac{w}{1+w}$.

    Consider the edges between the vertices of some (not the last) set $V_k$ of size $n$ in the root-layering and any vertex $v\in V_{k+1}$, and assume that we have not yet decided on the presence of these edges in the generated graph. Fix some subset $U\subseteq V_k$ of size $j$. Now let $\hat{\G}$ be the class of those graphs from $\G$ in which $U$ is the set of all vertices from $V_k$ connected to the vertex $v$ by an edge. By formula \eqref{eq:ggf_trick}, the probability that the generated graph is one of the graphs from $\hat{\G}$ is $\frac{\ggf{\hat{G}}(z,w)}{\ggf{G}(z,w)}=\frac{w^j}{(1+w)^n-1}$. Again, this probability is independent of the edges already generated. Moreover, $\frac{w^j}{(1+w)^n-1}$ is the probability of obtaining a fixed combination of $j$ successes in $n$ Bernoulli trials with a probability of success $\frac{w}{1+w}$, conditional on at least one success occurring. Therefore, we can sample the edges from the set $V_k$ to the vertex $v$, each with a probability $\frac{w}{1+w}$ until some edge is present.
\end{proof}

The Algorithm~\ref{algo:boltz:root_layering} presents the complete process of generating a directed acyclic graph from the Boltzmann model using root-layering decomposition.

\begin{algorithm}[!ht]
    \DontPrintSemicolon
    \SetKwFunction{GDAG}{$\Gamma\mathcal{DAG}_1$}
    \SetKwProg{Fn}{function}{}{}%

    \Fn{\GDAG{$z,w$}}{draw $n$ according to the distribution $\PP[n]=\frac{z^{n}}{(1+w)^{\binom{n}{2}}n!}\Set\left(-\frac{z}{(1+w)^n},w\right)$ \label{line:boltz:root_layering:1}\;
    create a set $U$ of $n$ vertices of a DAG\;
    $(V, E) \gets (\emptyset, \emptyset)$\;
    \While{$n\neq 0$}{
        draw a number $k$ according to the distribution $\PP[k|n]=\frac{\left((1-(1+w)^{-n})z\right)^{k}}{(1+w)^{\binom{k}{2}}k!}\frac{\Set(-(1+w)^{-k}z,w)}{\Set(-(1+w)^{-n}z,w)}$ \label{line:boltz:root_layering:2}\;
        create a set $W$ of $k$ vertices of a DAG\;
        \For{$v\in W$}{
            \For{$u\in V$}{
                add the edge $(u,v)$ to $E$ with probability $\frac{w}{1+w}$\;
            }
            \Repeat{$A\neq\emptyset$}{
                $A=\emptyset$\;
                \For{$u\in U$}{
                    add the edge $(u,v)$ to $A$ with probability $\frac{w}{1+w}$\;
                }
            }
            $E\gets E\cup A$\;
        }
        $(V,U,n)\gets(V\cup U,W,k)$\;
    }
    randomly relabel vertices in $V$\;
    \KwRet{$(V,E)$}\;}
    \caption{Boltzmann sampler for DAGs based on the root-layering}%
    \label{algo:boltz:root_layering}
\end{algorithm}
\begin{lemma}[Correctness of Algorithm~\ref{algo:boltz:root_layering}]\label{lem:root_layering:ok}
    Algorithm~\ref{algo:boltz:root_layering} is a Boltzmann sampler for directed acyclic graphs.
\end{lemma}
\begin{proof}
    It follows directly from the Lemmas~\ref{lem:root_layering_homogenous_markov_chain}-\ref{lem:root_layering_edges_distribution}.
\end{proof}

\begin{lemma}[Complexity of Algorithm \ref{algo:boltz:root_layering}]
  In order to produce a directed acyclic graph with~$n$ vertices, Algorithm~\ref{algo:boltz:root_layering} performs a quadratic number of calls to the random number generator and calculates the value of $\Set$ function $O(n)$ number of times.
\end{lemma}

\begin{proof}
For each edge, we draw a constant number of Bernoulli variables in expectation:
\begin{itemize}
    \item in the case of edges between vertices that lie in non-consecutive sets of root-layering, exactly one Bernoulli variable is drawn,
    \item in the case of edges between vertices that lie in consecutive sets $V_k$ and $V_{k+1}$ of root-layering, the number of Bernoulli variables drawn (conditioned on the cardinality of $V_k$) follows a geometric distribution.
\end{itemize}
Since there are $\binom{n}{2}$ possible edges in a DAG with $n$ vertices, generating edges takes a quadratic number of RNG calls on average.  

Following \cite[\S 5]{BoltzSamp2004}, we assume that the random variable with the outcome $k$ is drawn with the $O(k)$ number of real-arithmetic operations.
Since the generated DAG has $n$ vertices, the sum of the sizes of the generated layers is $n_1 + n_2 + \ldots = n$, and the total number of operations performed at line \ref{line:boltz:root_layering:1} and in all iterations of the while loop at line \ref{line:boltz:root_layering:2} by the RNG (which is of the same order as the total number of calls to the $\Set()$ function) is of $O(n_1 + n_2 + \ldots) = O(n)$.\footnote{In fact, it can be shown that the probability $\PP[k|n]$ at line $6$ is bounded by $\frac{c \cdot d^k}{k! (1+w)^{\binom{k}{2}}}$ for suitable constants $c$ and $d$. This implies that the expected size of the number of vertices in a layer is bounded by a constant.}\qedhere


\end{proof}

\subsection{Boltzmann sampling via a peeling process}

We now present an alternative approach to DAG sampling.
This second approach is based on a new recursive decomposition that can be expressed in terms of graphic generating functions without resorting to inclusion-exclusion.
A consequence of this property is that we can use standard Boltzmann techniques
in order to implement the sampler.

\subsubsection{Peeling process}

DAGs can be decomposed recursively via a peeling process: consider a specific
source of the DAG and split the DAG into two parts: the DAG induced by the
vertices accessible from our distinguished source and the DAG induced by the
other vertices.
Metaphorically, we ``peel-off'' the part of the DAG that is not accessible from
the distinguished source.
This is illustrated in Figure~\ref{fig:spec:peeling}, and the meaning of the formulas in the figure is provided in the combinatorial proof of Theorem~\ref{thm:spec:peeling} below.

\begin{theorem}[Peeling process]\label{thm:spec:peeling}
  The graphic generating function of directed acyclic graphs satisfies
  \begin{equation}\label{eq:peeling}
    u \partial_u \DAG(z, w, u)
    = uz \DAG\left(\frac{z}{1 + w}, w, u\right) \DAG\left(z, w, \frac{w}{1 + w}\right).
  \end{equation}
\end{theorem}

\begin{proof}
  Straightforward calculations yield
  \begin{equation*}
    u\partial_u \DAG(z, w, u)
    = uz \frac{\Set\left(\frac{(u-1)z}{1+w}, w\right)}{\Set(-z, w)}
    = uz \frac{\Set\left(\frac{(u-1)z}{1+w}, w\right)}{\Set\left(-\frac{z}{1+w}, w\right)}
         \frac{\Set\left(\frac{-z}{1+w}, w\right)}{\Set(-z, w)}.
  \end{equation*}
  And we conclude by observing that~$- \frac{1}{1+w} = \big(\frac{w}{1+w}
  -1\big)$.
\end{proof}

Of course, Theorem~\ref{thm:spec:peeling} would be of little interest without a combinatorial
interpretation. We provide an alternative proof here that highlights its
combinatorial meaning.

\begin{proof}[Combinatorial proof of Theorem~\ref{thm:spec:peeling}]
  Let~$G$ be a DAG with a marked source~$v$. It can be decomposed as follows.
  \begin{itemize}
    \item Denote by~$G_2$ the DAG induced by all the vertices that are reachable
      from the marked source (by following the directed edges), excluding the
      marked source~$v$.
    \item Denote by~$G_1$ the DAG induced by all the other vertices excluding~$v$.
    \item The initial (marked) DAG is obtained as a special instance of the
      arrow product of the three graphs~$\{v\}$,~$G_1$, and~$G_2$ where some edges are forced and some are forbidden.
      More specifically,~$G$ is obtained by performing the arrow product between~$G_1 \cup \{v\}$ and~$G_2$, while ensuring that there is an edge from~$v$ to every source of~$G_2$. This way there is no edge between~$v$ and~$G_1$ and all the vertices of~$G_2$ are reachable from~$v$.
  \end{itemize}
  The composition~$z=\frac{z}{1+w}$ in the product~$uz \DAG\left(\frac{z}{1+w}, w, u\right)$ captures the fact that there is no edge between the isolated vertex~$v$ (modelled by~$uz$) and the~$G_1$ component. More eloquently, the division of~$z$ by~$(1+w)$ means that we are ``removing'' one optional edge (specified by~$(1+w)$) for every vertex in the~$G_1$ component, which corresponds exactly to the optional edges encoded by the arrow product between~$uz$ and~$\DAG(z, w, u)$.
  Similarly, the composition~$u=\frac{w}{1+w}$ in the rightmost term of the equation represents the fact that for every source of the~$G_2$ component, we must have an edge coming from~$v$. Here again, the division by~$(1+w)$ inside the third argument of~$\DAG$ represents that we are ``removing'' one optional edge for each source, but then we multiply by~$w$. So, at the level of the specification, we replace every optional edge from~$v$ to a source by an obligatory one.
\end{proof}
\begin{figure}[htb]
  \centering
  \includegraphics[scale=1]{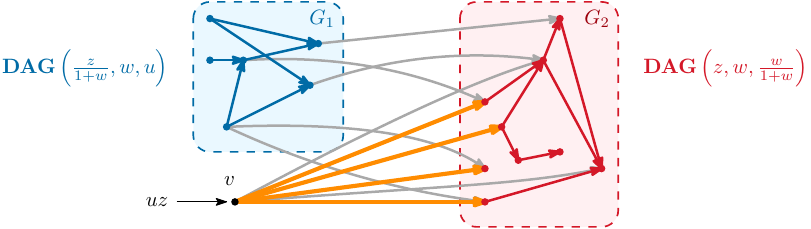}
  \caption{The ``peeling'' recursive decomposition of directed acyclic graphs.
  The labels are omitted.
  The marked source is the black vertex with the~$v$ label.
  Among the edges going from left to right, those starting from the marked
  source and going to a source of the graph on the right are always present and
  are depicted in thick orange, while all others are optional (depending on the
  initial DAG) and are depicted by thinner grey lines.}%
  \label{fig:spec:peeling}
\end{figure}

From the combinatorial interpretation provided above, we can now reformulate equation~\ref{eq:peeling} in a simple integral form:
\begin{equation}\label{eq:peelwithH}
  \DAG(z, w, u) = 1 + \ggf H(z, w, u) \DAG\bigl(z, w, \tfrac{w}{1+w}\bigr)
    \quad\!\!
    \text{with}
    \quad\!\!
    \ggf H(z, w, u) = z \!\!\int_0^u \!\!\!\DAG\bigl(\tfrac{z}{1+w}, w, t\bigr) dt.
\end{equation}
The effect of the integration on the combinatorial interpretation of this formula is that instead of considering arbitrarily pointed sources as the starting point of our decomposition, we choose the smallest of the sources.
This is a classical application of the box product from~\cite[\S II.6.3]{FS2009}, also used in~\cite{BFS1992} to enumerate increasing trees and in~\cite{BRS2012} to design Boltzmann samplers for them.
Combinatorially, the~$\mathcal H$-structures counted by the~$\ggf H$ function are the DAGs that can be obtained as the union of~$G_1$ and the isolated source~$v$ in the peeling process (see Figure~\ref{fig:spec:peeling}), but only when~$v$ has the smallest labels of all sources.

\begin{definition}
    An~$\mathcal H$-structure is a directed acyclic graph with a distinguished isolated source that has the smallest label of all the sources in the graph.
\end{definition}
The generation of~$\mathcal H$-structures is a central component of our peeling-based sampler.

\subsubsection{Boltzmann sampler}

First, by adapting the technique described in~\cite{BRS2012} for our use-case,
we can describe how to obtain a Boltzmann sampler of~$\mathcal H$-structures
from a DAG sampler, here denoted by~$\Gamma\mathcal{DAG}$.
In a second step, we will demonstrate how to perform the converse and thus define the two samplers in a mutually recursive fashion.

\begin{algorithm}[H]
  \DontPrintSemicolon
  \SetKwFunction{GH}{$\Gamma\mathcal H$}%
  \SetKwFunction{GDAG}{$\Gamma\mathcal{DAG}$}
  \SetKwProg{Fn}{function}{}{}%

  \Fn{\GH{$z, w, u$}}{
    $x \gets \text{uniform real number in~$[0; 1)$}$ \;
    compute the unique~$0 \le t \le u$ such that~$\frac{\Set((t-1)z, w) - \Set(-z, w)}{\Set((u-1)z, w) - \Set(-z, w)} = x$\label{line:boltz:H:solve} \;
    $G_1 \gets \text{\GDAG{$\frac{z}{1+w}, w, t$}}$\label{line:boltz:H:rec} \;
    $H \gets G_1 \cup \{\text{a fresh isolated vertex~$v$}\}$ \;
    choose a uniform label for~$v$ and relabel~$G_1$ accordingly \;
    perform a cyclic permutation of the labels of the sources of~$H$ so that~$v$
    has the smallest one \label{line:boltz:H:cycle} \;
    \KwRet{$H$}
  }
  \caption{Boltzmann sampler of~$\mathcal H$-structures}%
  \label{algo:boltz:H}
\end{algorithm}

\begin{lemma}[Correctness of Algorithm~\ref{algo:boltz:H}]\label{lem:H:ok}
  If~$\GDAG$ implements the Boltzmann model on directed acyclic graphs, then Algorithm~\ref{algo:boltz:H} implements the Boltzmann model on~$\mathcal H$. More precisely, 
  \begin{equation*}
    \PP\left[\GH{z, w, u} = H\right] = \frac{z^{\nbv H} w^{\nbe H} u^{\nbs H}}{n! (1+w)^{\binom n 2} \ggf H(z, w, u)}\cdot
  \end{equation*}
\end{lemma}
In fact, as shall be clear from the proof below, it suffices that~$\GDAG$ implements the Boltzmann model correctly on DAGs of size~$(n-1)$ to show that~$\GH$ produces~$\mathcal H$-structures of size~$n$ with the correct distribution.
This observation will allow us to use Lemma~\ref{lem:H:ok} as a template for proving the correctness of Algorithm~\ref{algo:boltz:peel} which is defined in a mutually recursive fashion with Algorithm~\ref{algo:boltz:H}.

\begin{proof}
  Consider an~$\mathcal H$-structure~$H$, denote by~$v$ its distinguished isolated source, and let~$n$, $m$, and~$k$ denote its number of vertices, edges and sources.
  Because of the way we make~$v$ the source with the smallest label at line~\ref{line:boltz:H:cycle}, if~$H$ has~$k$ sources in total, there are exactly~$k$ ways that Algorithm~\ref{algo:boltz:H} can produce~$H$, one for each cyclic permutation on its sources%
\footnote{This cycle trick ensures that~$v$ is the source with the smallest
label, while avoiding any bias. This is reminiscent from the cycle lemma used,
for instance, for Dyck path enumeration~\cite{DM1947}}.

  Let~$H^{(1)}, H^{(2)}, \ldots, H^{(k)}$ denote the~$k$ graphs that can be obtained by performing a cyclic permutation of the labels of the sources of~$H$, and for all~$1 \leq i \leq k$, let~$G_1^{(i)}$ denote the graph obtained by removing the distinguished isolated source from~$H^{(i)}$.
  
  In order to produce exactly~$H$, the function~$\Gamma\mathcal H(z, w, u)$ thus has to draw one of the~$G_1^{(i)}$ at line~\ref{line:boltz:H:rec} and then choose the unique label (with probability~$1/n$) that allows going from~$G_1^{(i)}$ to~$H$ after the cyclic relabelling.
  This probability can be expressed as
  \begin{equation*}
    \PP[\GH{z, w, u} = H] = 
    \sum_{i=1}^k \int_0^1 \PP\left[\GDAG\left(\frac z {1+w}, w, t(x)\right) = G_1^{(i)}\right] \frac 1 n dx
  \end{equation*}
  where~$t(x)$ is defined as in line~\ref{line:boltz:H:solve} of the algorithm.
  Assuming that~$\GDAG$ implements the Boltzmann distribution, the above formula can be rewritten as
  \begin{equation*}
      \frac{z^{n-1} w^m}{(n-1)! (1+w)^{\binom{n-1}2 + (n-1)}} \sum_{i=0}^k
      \int_0^1 \frac{t(x)^{k-1} dx}{\DAG\left(\frac{z}{1+w}, w,
      t(x)\right)}\cdot
  \end{equation*}
  By the change of variable~$t = t(x)$, we get that~$dx = \frac{z \Set\left((t-1)\frac{z}{1+w}, w\right) dt}{\Set((u-1)z, w) - \Set(-z, w)}$
  so that we have
  \begin{equation*}
    \PP[\GH{z, w, u} = H] = \frac{z^n w^m u^k}{n! (1+w)^{\binom n 2}} \frac{\Set(-\frac{z}{1+w}, w)}{\Set((u-1)z, w) - \Set(-z, w)}
  \end{equation*}
  and it can be checked, for instance, using equation~\eqref{eq:set_first_derivative} for integrating the~$\Set$ function, that
  \begin{equation*}
    \ggf H(z, w, u) = \frac{\Set((u-1)z, w) - \Set(-z, w)}{\Set(-\frac{z}{1+w}, w)}
  \end{equation*}
  which concludes the proof.
\end{proof}

Equipped with a template for sampling~$\mathcal H$-structures, it is now straightforward to derive a sampler for DAGs from equation~\eqref{eq:peelwithH}, which we describe in Algorithm~\ref{algo:boltz:peel}.
Note that~$\GH$ must use~$\Gamma\mathcal{DAG}_2$ internally here, so that both
functions are mutually recursive.

\begin{algorithm}[H]
  \DontPrintSemicolon
  \SetKwFunction{Bern}{Bernoulli}%
  \SetKwFunction{GDAG}{$\Gamma\mathcal{DAG}_2$}%
  \SetKwProg{Fn}{function}{}{}%

  \Fn{\GDAG{$z, w, u$}}{
    \eIf{\Bern{$1 / \DAG(z, w, u)$}\label{line:boltz:peel:bern}}{
      \KwRet{the empty graph} \;
    }{
      $H \gets \text{\GH{$z, w, u$}}$ \label{line:boltz:peel:rec1}\;
      $G_2 \gets \text{\GDAG{$z, w, \frac{w}{1+w}$}}$\label{line:boltz:peel:rec2} \;
      $G \gets H \cup G_2$, relabel accordingly \;
      \ForEach{pair of vertices~($v_1, v_2$) in~$H \times G_2$}{
        \eIf{$v_1$ is the distinguished source of~$G$ and~$v_2$ is a source of~$G_2$}{
          add an edge from~$v1$ to~$v_2$ \;
        }{
          add an edge from~$v_1$ to~$v_2$ with probability~$w / (w + 1)$ \;
        }
      }
      \KwRet{$G$}
    }
  }
  \caption{Boltzmann sampler for DAGs based on the peeling process}%
  \label{algo:boltz:peel}
\end{algorithm}

\begin{lemma}[Correctness of Algorithm~\ref{algo:boltz:peel}]\label{lem:peel:ok}
  Algorithm~\ref{algo:boltz:peel} is a Boltzmann sampler for directed acyclic graphs.
\end{lemma}
\begin{proof}
  For any~$z, w, u > 0$ such that~$\DAG(z, w, u)$ converges, and for any DAG~$G$, we denote by~$\PP_{z,w, u}[G]$ the probability that Algorithm~\ref{algo:boltz:peel} terminates and returns~$G$.
  We prove by induction on~$n$ that for all~$G \in \mathcal{DAG}$ with~$n$
  vertices and whenever this is well-defined,~$\PP_{z,w,u}[G] =
  {(1+w)}^{-\binom n 2} \frac{z^n w^{\nbe G} u^{\nbs G}}{n!\DAG(z, w, u)}$.
  The termination of the algorithm follows from the fact that these probabilities sum to~$1$.

  \begin{description}
    \item[Base case] We have that~$\PP_{z,w,u}[\emptyset] = 1 / \DAG(z, w, u)$.
    \item[Induction] Let~$G$ be a DAG with~$n > 0$ vertices, including~$k$
      sources and~$m$ edges. Let~$v$ denote the source of the smallest label
      in~$G$, and let~$G_2$ denote the graph induced by the vertices reachable
      from~$v$ (excluding~$v$). The graph induced by all other vertices
      (including~$v$) thus forms an~$\mathcal H$-structure~$H$.

      Assuming that the call to~$\GH$ at line~\ref{line:boltz:peel:rec1}
      and the recursive call at line~\ref{line:boltz:peel:rec2} return exactly~$H$ and~$G_2$, the
      probability that the subsequent lines produce exactly~$G$ is
      \begin{equation*}
        P_{\text{finish}} \coloneqq
        \frac{\nbv{H}! \nbv{G_2}!}{n!} \cdot
        \frac{w^{m - \nbe{H} - \nbe{G_2} - \nbs{G_2}}}%
             {{(1+w)}^{\nbv{H}\nbv{G_2} - \nbs{G_2}}}
      \end{equation*}
      where the first fraction accounts for the probability of choosing the
      right relabelling, and the second fraction quantifies the probability of
      drawing the correct edges from~$H$ to~$G_2$ in order to obtain exactly~$G$\footnote{Note that this is reminiscent of equation~\eqref{eq:ggf_trick} on page~\pageref{eq:ggf_trick}}.
      By induction, the recursive call to~$\GDAG$ hidden inside~$\GH$ at line~\ref{line:boltz:peel:rec1} and the recursive call at line~\ref{line:boltz:peel:rec2} implement the Boltzmann model on DAGs, so that~$\GH$ implements the Boltzmann model on~$\mathcal H$-structures, and we have
      \begin{align*}
        \PP_{z,w,u}[G]
        &=\left(1 - \frac{1}{\DAG(z, w, u)}\right)
          \PP_{\mathcal H, z, w, u}[H] \cdot
          \PP_{z,w, \frac{w}{1+w}}[G_2] \cdot
          P_{\text{finish}}
        \\
        &=\left(1 - \frac{1}{\DAG(z, w, u)}\right)
          \frac{z^n w^m u^k}{n! (1+w)^{\binom n 2}} \frac{1}{\ggf H(z, w, u) \DAG(z, w, \frac w {1+w})}
        \\
        &= \frac{z^n w^m u^k}{n! (1+w)^{\binom n 2} \DAG(z, w, u)}\cdot\qedhere
      \end{align*}
  \end{description}
\end{proof}

We express the complexity of Algorithm~\ref{algo:boltz:peel} in terms of the number of calls to the random number generator and the number of times we need to solve the equation inside the call to~$\GH$. The rest of the algorithm only consists of building the output graph, which is clearly linear in the size of the output\footnote{Note that if the output has~$n$ vertices, it has~$m = \Theta(n^2)$ edges, so ``linear'' actually means~$O(m)$ in this context.}.
\begin{lemma}[Complexity of Algorithm \ref{algo:boltz:peel}]\label{lem:peel:complexity}
  In order to produce a directed acyclic graph with~$n$ vertices, Algorithm~\ref{algo:boltz:peel} performs a quadratic number of calls to the random number generator and solves the system at line~\ref{line:boltz:H:solve} in Algorithm~\ref{algo:boltz:H} exactly~$n$ times.
\end{lemma}

\begin{proof}
  For each vertex that we build, we have to draw a Bernoulli variable beforehand (at line~\ref{line:boltz:peel:bern} and then a uniform real number~$x$ before solving the equation at line~\ref{line:boltz:H:solve} in Algorithm~\ref{algo:boltz:H}).
  This incurs exactly~$2n$ calls to the RNG and~$n$ solves.
  Moreover, we have to draw several Bernoulli variables in order to decide which edges to add or not in the \textbf{\texttt{foreach}} loop.
  This incurs at most~$\binom n 2$ RNG calls: one for each possible pair of vertices.
\end{proof}

\section{Exact size sampling}\label{sec:exactsize}

Both approaches presented above are amenable to quadratic-time exact-size sampling via rejection. Using the rejection principle to derive exact-size samplers is natural in the context of Boltzmann sampling, though it is generally inefficient.
In our case, we can still achieve good performance due to a shared feature of the two algorithms: we can postpone the generation of the edges.
An alternative approach to exact-size sampling that is permitted by the peeling process is to resort to a trick known as ``leapfrogging''.
This technique, described in \cite[\S 7.1]{BoltzSamp2004}, leverages the properties of super-critical sequences in order to derive a fast algorithm in this specific case.
Both approaches are presented below.

\subsection{Rejection sampling}\label{sec:rejection}

Here, the rejection method involves generating DAGs until a graph with the desired number of vertices is obtained. We can tune our algorithms to be more efficient for such generation. First, the only process we need to repeat is the generation of a graph skeleton without edges, as this is the part that establishes the number of vertices. In cases where we already know that the number of vertices will be greater than what we require, we can interrupt the process of generating the graph skeleton and start it over. Edge sampling occurs after skeleton sampling.

Fix the parameter $w$ (choose $w=1$ for the uniform generation of DAGs with the same number of vertices). Suppose that we want the generated graph to have exactly $n$ vertices, and we use the rejection method, \textit{i.e.} we choose a random graph until there are exactly $n$ vertices. For this method to be efficient, the probability of drawing a graph with $n$ vertices should be as high as possible. Hence, we are looking for the $z_n$ parameter that maximises the value
\[
\PP[\nbv{\Gamma\mathcal{DAG}(z_n,w)}=n]=[t^n]\frac{\DAG(z_nt,w)}{\DAG(z_n,w)}~.
\]
It is a well known fact in Boltzmann sampling that this parameter ensures that the expected size of the generated structures is equal to $n$, so $z_n$ can be calculated by solving the following formula
\[
n=\mathbb{E}[v(\Gamma\mathcal{DAG}(z_n,w))]=\frac{z_n \partial_z \DAG(z_n,w)}{\DAG(z_n,w)}~.
\]
The derivative of the function~$\DAG(z, w)$ can be easily computed from equations~\eqref{eq:dag} and~\eqref{eq:set_first_derivative}, which yields
\begin{equation}\label{eq:zn}
    n = \frac{z_n \partial_z \DAG(z_n, w)}{\DAG(z_n, w)} = \frac{z_n\,\Set\left(-z_n/(1+w),w\right)}{\Set(-z_n,w)}~.
\end{equation}
Solving the equation for $z_n$ can be accomplished using numerical methods.

Now, we would like to know how many iterations of the rejection sampler are needed, on average, to produce a random DAG with $n$ vertices. First, we need to obtain an asymptotic formula for $[z^n]\DAG(z,w)$. In \cite{wang2018zeros} it was shown that the function
$f(z)=\sum_{n\geq0}q^{\binom{n}{2}}\frac{z^n}{n!}$
for $q\in (0,1)$ has infinitely many zeros, all of which are real and negative. Let $\rho_w$ be the least zero of the function $f(-z)$ for $q=\frac{1}{1+w}$, which is the least zero of the function $\Set(-z,w)$ ($\rho_1 \approx 1.4880785$). Then $\rho_w$ is the radius of convergence of the function $\DAG(z,w)$. It can be shown that $\DAG(z,w)$ has a simple pole at $z = \rho_w$, and using Theorem~IV.10 from \cite{FS2009}, that
\begin{equation}\label{eq:dag_asymptotic}
    [z^n]\DAG(z,w)\overset{n\to\infty}{\sim}\frac{1}{\rho_w \Set(-\rho_w/(1+w),w)}\frac{1}{\rho_w^n}~.
\end{equation}
From \eqref{eq:set_first_derivative} we have
\[
\Set(-z_n,w)\overset{n\to\infty}{\sim}(z_n-\rho_w)\partial_z\Set(-z_n,w)=(\rho_w-z_n)\Set(-z_n/(1+w),w)~,
\]
and therefore, using~\eqref{eq:zn}, we obtain
\begin{equation}\label{eq:vertices_asymptotic}
    n=\frac{z_n\,\Set(-z_n/(1+w),w)}{\Set(-z_n,w)}\overset{n\to\infty}{\sim}\frac{\rho_w}{\rho_w-z_n}.
\end{equation}
Finally, from the above equations \eqref{eq:dag_asymptotic} and \eqref{eq:vertices_asymptotic}, we get
\begin{align*}
    &\PP[v(\Gamma\mathcal{DAG}(z_n,w))=n]=\frac{z_n^n[t^n]\DAG(t,w)}{\DAG(z_n,w)}\\
    &\overset{n\to\infty}{\sim}\frac{\Set(-z_n,w)}{\rho_w\Set(-\rho_w/(1+w),w)}\left(\frac{z_n}{\rho_w}\right)^n\overset{n\to\infty}{\sim}\frac{1}{n}\left(1-\frac{\rho_w-z_n}{\rho_w}\right)^{\frac{\rho_w}{\rho_w-z_n}}\overset{n\to\infty}{\sim}\frac{1}{e\cdot n}~.
\end{align*}
Hence, since the number of rejection sampler iterations has a geometric distribution with a probability of success asymptotically $\frac{1}{e\cdot n}$, the average number of these iterations approaches $e\cdot n$, making it linear with respect to $n$.

\subsection{Leapfrogging}

We proved that rejection sampling with our Boltzmann samplers achieves~$O(n^2)$ exact size sampling.
We can actually go further and optimise the constant hidden in the big-O using a technique called leapfrogging.

\subsubsection{Leapfrogging in a nutshell}

In the paper~\cite{BoltzSamp2004} where the framework of Boltzmann sampling is
introduced for the first time, the authors note in Section~7.1 a seemingly
anecdotal case where the method proves to be especially powerful.
When the combinatorial objects to be sampled admit a
specification of the form~$\A = \Seq\B$ where the generating function~$B(z)$
of~$\B$ is larger than~$1$ near its dominant singularity, the sequence is called \emph{supercritical}.
In this case, the generating function~$A(z)$ of~$\A$ admits a simple pole~$\rho$
where~$B(\rho) = 1$ and $B(z)$ is analytic in a neighbourhood of~$\rho$.
This allows us to show that a large~$\A$ structure in the Boltzmann model is composed
of a long sequence of small~$\B$ structures; however, it goes further than this.
Indeed, although the Boltzmann model is not well defined at~$z=\rho$
(because~$A(\rho) = \infty$), we can still define an early-interrupt
``critical'' Boltzmann sampler by generating a sequence of~$\B$ structures with
parameter~$\rho$ and halting as soon as a target size is attained.
The authors refer to this as the ``leapfrogging'' principle.
This is described in Algorithm~\ref{algo:genericleapfrog}.

\begin{algorithm}
  \DontPrintSemicolon
  \SetKwFunction{Bern}{Bernoulli}%
  \SetKwFunction{Frog}{$\Gamma_{\textrm{frog}}$}%
  \SetKwFunction{BoltzB}{$\Gamma\B$}%
  \SetKwProg{Fn}{function}{}{}%

  \Fn{\Frog{$n$}}{
    \Repeat{until~$\sum_{b \in S}|b| = n$}{
      $S \gets \textrm{empty sequence}$ \;
      \While{$\sum_{b \in S} |b| < n$}{
        $b \gets \text{\BoltzB{$\rho$}}$ \;
        append $b$ to $S$ \;
      }
      }
    \KwRet{$S$} \;
  }
  \caption{The leapfrogging algorithm for~$\A = \Seq\B$}%
  \label{algo:genericleapfrog}
\end{algorithm}

Note that the original paper~\cite{BoltzSamp2004} has a slightly more general
presentation that enables approximate-size sampling, however, we chose to focus on exact-size sampling only here.
Algorithm~\ref{algo:genericleapfrog} has a remarkable property.
Since the ``leaps'' --- the~$\B$-structures generated at each iteration --- are
small, the probability that the generated sequence inside the \texttt{\textbf{repeat}} block has size exactly~$n$ is actually a non-zero constant, even when~$n \to \infty$. As a consequence, the algorithm succeeds in generating an~$\A$ structure of size~$n$ after only~$O(1)$ rejections.
That is the idea that we leverage here to achieve optimal exact-size sampling of DAGs.

\subsubsection{A sequential specification of DAGs}

Recall that the peeling process described in Theorem~\ref{thm:spec:peeling} peels off an~$\mathcal H$-structure (by removing all vertices that are not reachable from the smallest source) and leaves us with a smaller DAG.
We can actually repeat this process on the resulting DAG, which yields the following alternative formulation.
\begin{lemma}[Sequential decomposition]\label{lem:spec:sequence}
  The graphic generating function of directed acyclic graphs satisfies
  \begin{equation*}
    \DAG(z, w, u) = 1 + \frac{\ggf H(z, w, u)}{1 - \ggf H\left(z, w, \frac{w}{1+w}\right)}
  \end{equation*}
\end{lemma}
\begin{proof}
    When substituting~$u=\frac{w}{1+w}$ in~\eqref{eq:peelwithH}, we get 
    \begin{equation*}
        \DAG\left(z, w, \frac{w}{1+w}\right) = 1 + \ggf H\left(z, w, \frac{w}{1+w}\right) \DAG\left(z, w, \frac{w}{1+w}\right)
        = \frac{1}{1 - \ggf H\left(z, w, \frac{w}{1+w}\right)}\cdot\qedhere
    \end{equation*}
\end{proof}
This new expression allows the familiar pseudo-inverse operation on series to appear, which is naturally interpreted as a \emph{sequence} of~$\mathcal H$-structures.
It is worth noting, though, that this sequence differs slightly from the usual notion of a sequence: here, there is an arrow product between each~$\mathcal H$-structure and all the following ones, and some of the edges are forced.
This sequential decomposition, or iterated peeling process, is illustrated in Figure~\ref{fig:spec:sequence}.

\begin{figure}[htb]
  \centering
  \includegraphics[width=0.8\textwidth]{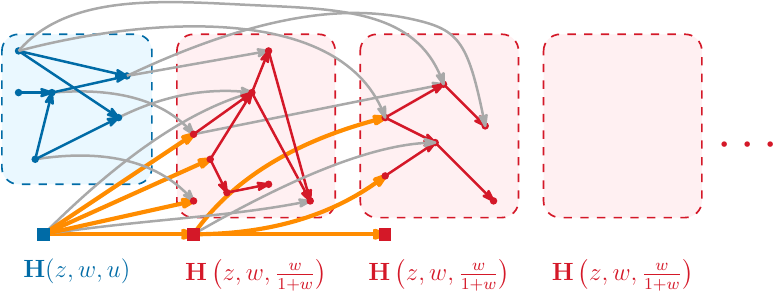}
  \caption{Iterated peeling process. Labels are omitted. At each decomposition
  step, the smallest of the sources is selected and the vertices that are not
  accessible from this source are peeled off, as well as the selected source.
  This process is iterated until the remaining DAG is empty. The~$\mathcal H$ structures are the union of the DAG component with the isolated vertex at each step.}
  \label{fig:spec:sequence}
\end{figure}

\subsubsection{Leapfrogging for DAGs}

The sequential decomposition from Lemma~\ref{lem:spec:sequence} is amenable to leapfrogging because of the following observation.
Regardless of the value of~$u \in (0;1]$, the function~$\ggf H(z, w, u)$ is analytic in the domain~$|z| < (1 + w)\rho_w$, which contains~$\rho_w$.
It follows that the sequence is super-critical and that~$\ggf H\bigl(\rho_w, w, \frac{w}{1+w}\bigr) = 1$.
Although the sequence in our case is not a sequence in the same sense as in~\cite{BoltzSamp2004}, the same ideas apply, and one can obtain a leapfrogging algorithm for DAGs by drawing sequences of~$\mathcal H$-structures until we find one of total size exactly~$n$. Note that in our case, all leaps but the first one must be drawn with parameter~$u=\frac{w}{1+w}$.

Note that the leapfrogging idea is generic in terms of which algorithm is used to generate the leaps, as long as it implements the Boltzmann model on~$\mathcal H$. In Algorithm~\ref{algo:leapfrog},~$\GH$ can use either of our two algorithms under the hood to generate DAGs.

\begin{algorithm}
  \DontPrintSemicolon
  \SetKwFunction{UnifDAG}{UnifDAG}%
  \SetKwProg{Fn}{function}{}{}%

  \Fn{\UnifDAG{$n$}}{
    \Repeat{$\sum_{H \in S} \nbv H \neq n$}{
      $S \gets \text{new sequence containing one graph drawn from \GH{$\rho_1, 1, 1$}}$ \;
      \While{$\sum_{H \in S} \nbv H < n$}{
        append a new~\GH{$\rho_1, 1, \tfrac 1 2$} to~$S$ \;
      }
    }
    shuffle the labels of the different leaps \label{line:frog:shuffle} \;
    add edges between each leap and its following leaps as in Algorithm~\ref{algo:boltz:peel}\label{line:frog:edges} \;
    \KwRet{the resulting DAG}
  }
  \caption{Exact-size sampler of DAGs. The special symbol~$\bot$ symbolises a failure.}%
  \label{algo:leapfrog}
\end{algorithm}
An important optimisation here is that we wait to find a correct sequence of~$\mathcal H$-structures before actually connecting the various components together.
This allows to lower the complexity of this algorithm because most of the computational cost lies in this second phase: it generates~$\Theta(n^2)$ edges by drawing about~$\frac{n^2}{2}$ Bernoulli variables and performing~$\Theta(n^2)$ memory accesses to store those edges.

\begin{lemma}[Correctness of Algorithm~\ref{algo:leapfrog}]\label{lem:leapfrog:ok}
  The function \UnifDAG from Algorithm~\ref{algo:leapfrog}, when called with parameter~$n$ returns a uniform directed acyclic graph with~$n$ vertices.
\end{lemma}

\begin{proof}
  We first express the probability~$\PP_{\text{while}}[S]$ that a given sequence~$S = (H_1, H_2, \ldots, H_j)$ of~$\mathcal H$-structures, of total size~$n > 0$, is produced by the \textbf{\texttt{while}} loop. By the definition of the Boltzmann model on~$\mathcal H$, this probability is given by
  \begin{equation*}
    \PP_{\text{while}}[S] = \frac{\rho_1^{n_1}}{n_1! 2^{\binom {n_1} 2} \ggf H(\rho_1, 1, 1)}
    \prod_{i=2}^j \frac{\rho_1^{n_i} (\tfrac 1 2)^{k_i}}{n_i! 2^{\binom {n_i} 2} \ggf H(\rho_1, 1, \tfrac 1 2)}
    = \frac{\rho_1^n 2^{-\sum_{i=2}^j k_i - \sum_{i=1}^j \binom {n_i} 2}}{\ggf H (\rho_1, 1, 1) \prod_{i=1}^j n_i!}
  \end{equation*}
  where~$n_i$ and~$k_i$ denote the number of edges and sources of the~$i$-th term of the sequence.
  Of course, the \texttt{\textbf{while}} loop can fail to produce a sequence of total size~$n$, so the probability $\PP_{\text{repeat}}[S]$ that the \texttt{\textbf{repeat}} block yields a specific sequence~$S$ by rejection is given by
  \begin{equation*}
    \PP_{\text{repeat}}[S]
    = \frac{\PP_{\text{while}}[S]}{\sum_{\text{$S'$ of total size~$n$}} \PP_{\text{while}}[S']}
    = \frac{\ggf H(\rho_1, 1, 1)\PP_{\text{while}}[S]}{\rho_1^n [z^n] \frac{\ggf H(z, 1, 1)}{1 - \ggf H(z, 1, \tfrac 1 2)}}
    = \frac{\ggf H(\rho_1, 1, 1) \PP_{\text{while}}[S]}{\rho_1^n [z^n] \DAG(z, 1, 1)}\cdot
  \end{equation*}
  Finally, in order for the leapfrogging algorithm to produce a specific DAG~$G$, it first has to produce the exact sequence~$S$ arising from the peeling process in the rejection phase, and then shuffle the labels and generate the edges in the unique way that corresponds to~$G$.
  This happens with probability
  \begin{align*}
    &
    \PP_{\text{repeat}}[S] \cdot \binom{n}{n_1, n_2, \ldots, n_j}^{-1}
    \cdot 2^{-\sum_{i + 2 \leq i'} n_i n_{i'}}
    \cdot 2^{-\sum_{i=1}^{j-1} n_i n_{i+1} - k_{i+1}} \\
    &= \frac{1}{n! 2^{\binom n 2}} \frac{1}{[z^n] \DAG(z, 1, 1)}\cdot\qedhere
  \end{align*}
\end{proof}

As discussed above, the leapfrogging approach enables the quick generation of a structure of exact size~$n$. In our case, because the majority of the computational cost of generation lies in the edges \emph{between} the layers, this implies that we obtain an optimal sampler.

\begin{lemma}[Complexity of Algorithm~\ref{algo:leapfrog}]
  In order to generate a uniform directed acyclic graph of size~$n$, Algorithm~\ref{algo:leapfrog} performs a~$O(n)$ number of system inversion (in the~$\GH$ function) and uses~$\frac{n^2}{2} + o(n^2)$ random bits in average.
\end{lemma}
\begin{proof}
  As in the proof of Lemma~\ref{lem:peel:complexity}, the number of system inversions is clearly equal to~$n$.
  From the proof of Lemma~\ref{lem:leapfrog:ok}, one can compute that the probability that the \textbf{\texttt{while}} loop successfully produces a sequence of total size~$n$ is
  \begin{equation*}
    \frac{\rho_1^n [z^n] \DAG(z, 1, 1)}{\ggf H(\rho_1, 1, 1)}
    = \rho_1^n [z^n] \DAG(z, 1, 1) \Set(-\tfrac{\rho_1}2, 1)
    \underset{n\to\infty}\to \rho_1^{-1} \approx 0.672.
  \end{equation*}
  Furthermore, because the sequence is super-critical, every leap in the sequence is, on average, of constant size.
  It follows that the total cost of generating the edges inside the leaps during the rejection phase is linear in~$n$.

  Since the final stage of the algorithm consists of generating~$\binom n 2 - O(n)$ Bernoulli variables with parameter~$\tfrac 1 2$, it costs exactly~$\binom n 2 - O(n)$ random bits, which dominates the cost of the rejection phase.
\end{proof}

\section{Implementation considerations}

Evaluating the~$\Set(z, 1)$ function with high precision is efficient in practice due to its fast convergence.
Our implementation~\cite{WGimpl} of the algorithms relies on floating-point arithmetic, which is generally assumed to be adequate for the practical implementation of Boltzmann samplers.
Of course, a more precise implementation with arbitrary precision that keeps track of numerical errors is also possible and will not impact the complexity of our samplers, as most of their computational cost lies in the generation of $\texttt{Bernoulli}(\tfrac 1 2)$ variables for connecting edges.

Furthermore, for both the rejection sampler and the leapfrogging algorithm, a fast implementation should allocate the DAG only once and perform the rejection phase in place to further minimise the cost of rejection.

As a final remark, it must be noted that Algorithm~\ref{algo:boltz:peel} is not tail-recursive. It is thus suitable for small size sampling, such as for the generation of the leaps in Algorithm~\ref{algo:leapfrog}, but it might not be appropriate for the rejection approach presented in Section~\ref{sec:rejection} in its current form.

\section{Conclusion and perspectives}

The present paper extends the framework of Boltzmann sampling to digraph families and showcases its effectiveness by providing Boltzmann samplers for DAGs. Building on top of these samplers, we also provide an optimal exact-size sampler for DAGs, thus closing the gap between available samplers and the theoretical complexity lower bound given by the entropy.
At a more fundamental level, we also present a new decomposition scheme that is amenable not only to random generation but also to analytic combinatorics techniques.

The reader may have noticed that another DAG decomposition might be obtained by differentiating the $\DAG(z,w,u)$ function with respect to $z$. Additionally, other decompositions that are not a direct consequence of the graphical generating function of DAGs (as in the case of root-layering decomposition) can be found. With the approach presented here, they can be used to build different DAGs samplers. However, it is worth noting that not all of them satisfy the sequential decomposition schema (see Lemma~\ref{lem:spec:sequence}), which is key to applying the leapfrogging idea to create an asymptotically optimal DAGs sampler.
Nevertheless, these alternative specifications are still worthy of interest. In particular, Boltzmann samplers have been used to study the probabilistic aspects of generated graph parameters (see~\cite{dmtcs:3552, 10.5555/2133036.2133127, 10.1112}).
Therefore, the next natural step would be to use the samplers presented in this paper, along with other samplers employing the aforementioned specifications, to obtain novel information about the probability distribution of various parameters of DAGs.

Finally, an analysis of the graphic generating functions for various classes of digraphs was presented in the work of Dovgal, de Panafieu, Ralaivaosaona, Rasendrahasina, Wagner~\cite{DdRRW2024}. 
Therefore, our graphic Boltzmann model and the approach presented for creating efficient random DAGs generators can be applied to build samplers for these other classes of digraphs.

\bibliographystyle{plain}
\bibliography{bibliography}

\end{document}